\documentclass[12pt,english]{article}

\usepackage[T1]{fontenc}
\usepackage[utf8]{inputenc}
\usepackage[english]{babel}
\usepackage[margin=1.25in]{geometry}
\usepackage{amsmath,amssymb,amsthm,mathtools,bm}
\usepackage{natbib}
\usepackage{enumitem}
\usepackage{microtype}
\usepackage{setspace}
\usepackage[colorlinks=true,linkcolor=blue,urlcolor=blue,citecolor=blue]{hyperref}
\usepackage{multirow}
\usepackage{array}

\DeclareMathOperator{\Var}{Var}

\usepackage{color}
\usepackage{xcolor}
\definecolor{red}{rgb}{1,0,0}
\definecolor{gred}{rgb}{0.88,0,0}
\definecolor{redish}{rgb}{0.8,0,0}
\definecolor{green}{rgb}{0,0.45,0}
\definecolor{llgreen}{rgb}{0.9,1,0.9}
\definecolor{ggreen}{rgb}{0,0.4,0}
\definecolor{blue}{rgb}{0,0,1}
\definecolor{gblue}{rgb}{0,0,0.88}
\definecolor{lblue}{rgb}{0.5,0,1}
\definecolor{grey}{rgb}{0.6,0.6,0.6}
\definecolor{lgrey}{rgb}{0.88,0.88,0.88}
\definecolor{orange}{rgb}{1, 0.6, 0}
\definecolor{gorange}{rgb}{0.8, 0.48, 0}
\definecolor{lorange}{rgb}{1, 1, 0.5}
\definecolor{purple}{RGB}{153, 50, 204}

\newcommand{\I}{\mathcal{I}}
\newcommand{\E}{{\mathbb{E}}}

\renewcommand{\S}{{\mathcal{S}}}

\usepackage{amsthm}
\newtheorem{Theorem}{Theorem}
\newtheorem{Note}{Note}
\newtheorem{Def}{Definition}

\newtheorem{Lemma}{Lemma}
\newtheorem{Proposition}{Proposition}
\newtheorem{Example}{Example}

\newcommand{\te}{\theta}
\newcommand{\Te}{\Theta}
\newcommand{\lt}{\left}
\newcommand{\rt}{\right}
\newcommand{\h}{\widehat}

\newcommand{\td}{\tilde}
\newcommand{\cE}{\mathcal{E}}
\newcommand{\x}{\mathsf{x}}

\title{Risk-minimizing GUE-implementations of truthful reporting in a Dynamic Social Choice Problem\thanks{We are thankful for the helpful comments of Péter Ernő Frenkel.}}
\author{Máté Viczián\thanks{{\scriptsize ELTE Faculty of Science, Institute of Mathematics, Budapest, Hungary. \textit{Email}: \href{viczi.mate@gmail.com}{viczi.mate@gmail.com}.}}, Endre Csóka\thanks{{\scriptsize Alfréd Rényi Institute of Mathematics, Budapest, Hungary; ELTE Faculty of Science, Institute of Mathematics, Budapest, Hungary. Supported by the NRDI grant KKP~138270.}}}

\title{Risk-minimizing GUE-implementations of truthful reporting in a Dynamic Social Choice Problem\thanks{We are thankful for the helpful comments of Péter Ernő Frenkel.}}
\author{Endre Csóka\thanks{{\scriptsize Alfréd Rényi Institute of Mathematics, Budapest, Hungary; ELTE Faculty of Science, Institute of Mathematics, Budapest, Hungary. Supported by the NRDI grant KKP~138270.}},
Máté Viczián\thanks{{\scriptsize ELTE Faculty of Science, Institute of Mathematics, Budapest, Hungary. \textit{Email}: \href{mailto:viczi.mate@gmail.com}{viczi.mate@gmail.com}.}}}
\date{}%October 3, 2026}%\date{February 2026}

\begin{document}

\maketitle

\begin{abstract}
    The Guaranteed Utility Mechanism (GUM) was introduced to implement efficiency in Guaranteed Utility Equilibrium (GUE) in dynamic social choice problems. We show that this mechanism also has a previously unknown connection to risk minimization. Fixing an efficient decision policy, we study variance-minimizing budget-balanced transfer rules in the class of GUE-implementations of truthful reporting.

    A symmetrized version of GUM (sym-GUM) minimizes the sum of the agents' utility variances under truthful reporting when there is no public state. When a public state is present, a slight modification of this rule yields the minimizer for the same objective. Sym-GUM, however, does not minimize transfer risk. The minimizer for that problem is instead the symmetrized version of a closely related GUM-like transfer rule. In both problems, the minimizing truthful total transfers are unique up to agent-specific constants summing to zero and events of probability zero.
\end{abstract}

\section{Introduction}
Guaranteed Utility Equilibrium (GUE) is a strong equilibrium concept with several properties that dominant-strategy equilibrium does not generally possess, including collusion-proofness and robustness to repeated play \citep{GUE}. In the transferable-utility setting, a GUE allows each agent to guarantee himself an expected utility independently of the strategies of the other agents, while the sum of these guarantees equals the maximum achievable expected total utility.

As shown in \citet{CLRT}, efficiency can be implemented in GUE in a rich dynamic social choice problem with Markovian private information. This contrasts with the classical Arrow--d'Aspremont--Gérard-Varet (AGV) mechanism \citep{arrow1979,AGV} and its dynamic extension, the Balanced Team Mechanism of \citet{AtSe}. These mechanisms implement efficiency in Bayesian Nash equilibrium, but they are not generally robust to weak dominance: iterated elimination of weakly dominated strategies may eliminate all efficient outcomes. By contrast, the Guaranteed Utility Mechanism (GUM) of \citet{CLRT} implements efficiency in GUE and therefore inherits the corresponding robustness and collusion-proofness properties.

The weakness of AGV may have gone unnoticed in part because the two-agent case is fundamentally different. In the two-agent static social choice problem, AGV itself implements efficiency in GUE, and the same is true of its dynamic extension, the Balanced Team Mechanism. Consequently, in the two-agent case every Bayesian Nash equilibrium of these mechanisms is efficient and utility-equivalent. With three or more agents, this GUE property is lost in general. Thus, examples with only two agents can conceal an important distinction between AGV-type mechanisms and GUM. \citep{GUE}\bigskip

GUE implementation leaves substantial freedom in the choice of transfers, which naturally raises the question of secondary design objectives. This paper studies risk minimization as such a secondary objective. Throughout, the decision policy is fixed and maximizes expected total payoff under truthful reporting. The design variable is therefore the budget-balanced transfer rule.

Two risk criteria are considered under truthful reporting. First, we minimize the sum of the agents' utility variances\footnote{We are not talking about minimizing the variances of the \emph{per period} utilities of the agents. By the utility of an agent, we mean the sum of his per period utilities.}. The minimizer is a modification of symmetrized GUM which, in addition to symmetrizing the effects of private reports, shares equally the change in anticipated total payoff caused by the public-state realization. Second, we minimize the sum of the variances of the agents' total transfers\footnote{We are not talking about minimizing the variances of per period transfers.}. For this objective, we introduce a GUM-like transfer rule and symmetrize it over all orders of the agents. We also show that the two minimizers are almost unique in total transfers modulo constant zero-sum transfers to the agents.\bigskip

In the two-player static case, the symmetrized GUM-like transfer rule \eqref{eq:sym-GUM-like} coincides with AGV. The two minimizer transfer rules, however, do already differ in the two-player static case with no public state. In the three-player case, all three mechanisms differ. The following example demonstrates all three facts.\footnote{This introductory example uses continuous types for simplicity; the formal results below are stated for finite type spaces.}

\begin{Example}[A simple allocation problem]
    There are $N$ agents. The agents need to decide which one of them gets a specific calf they each value differently. Agent $i$'s state $\te^i$ is a number chosen uniformly from the interval $[0,1]$; this is how much he values the calf. The agents each report their valuation of the calf, and then the one who said the largest number receives the calf (this is, of course, an efficient decision policy). The payoff of agent $i$ is zero if he does not receive the calf, and is $\te^i$ otherwise. We may assume some tiebreak system is used, but as ties occur with zero probability under truth-telling, this does not affect any of the variances.

    The following shorthands will come in handy:
    \begin{align}
        M^H=\max\{\h\te^i\}_{i\in H},\qquad v^i=\begin{cases}
            \h\te^i\text{, if }$i$\text{ gets the calf}\\ 0\text{ else}
        \end{cases}.
    \end{align}

    In this example, we denote the AGV transfer to agent $i$ by $y_{AGV}^i$, the utility variance minimizer GUE-implementation of truthful reporting \eqref{modifiedsymGUM} by $y_{U}^i$ and the transfer variance minimizing one \eqref{eq:sym-GUM-like} by $y_T^i$.

    \paragraph{The $N=2$ case.} Here,
    \begin{align}
        y_{AGV}^i=y_T^i=\frac{(\h\te^j)^2-(\h\te^i)^2}{2},
    \end{align}
    where $j\neq i$ represents the other agent. On the other hand,
    \begin{align}
        y_U^i=-\frac{y_{T}^i}{2}+\frac{v^j-v^i}{2}.
    \end{align}
    We now present the sum of utility variances $V_U=\sum_{i}\Var(U^i)$ and the sum of transfer variances $V_Y=\sum_i\Var(y^i)$ under truth-telling for the three transfer rules, rounded to three decimal places.
    \begin{center}
        \begin{tabular}{c|c c}
              & $V_Y$ & $V_U$ \\ \hline
    AGV=T-min \eqref{eq:sym-GUM-like} &$0.089$&$0.100$\\
    U-min \eqref{modifiedsymGUM}    &$0.139$&$0.050$
        \end{tabular}
    \end{center}
    So the transfer variance minimizer rule does in fact give a lower transfer variance, and utility variance with the utility variance minimizer is halved at the expense of transfer volatility.

    \paragraph{The $N=3$ case.} Here, all three transfer rules differ.
    \begin{align}
        y_{AGV}^i&=\frac{(\h\te^j)^3+(\h\te^k)^3-2(\h\te^i)^3}{3},\\
        y_T^i&=\frac{y_{AGV}^i}{2}+\frac{2(M^{jk})^2-(M^{ik})^2-(M^{ij})^2}{4},\\
        y_U^i&=-\frac{y_T^i}{3}+\frac{v^j+v^k-2v^i}{3},
    \end{align}
    where $i,j$ and $k$ are the three distinct agents. The sum of utility variances $V_U$ and the sum of transfer variances $V_Y$ under truth-telling for the respective transfer rules are as follows (rounded to three decimal places).
    \begin{center}
        \begin{tabular}{c|c c}
              & $V_Y$ & $V_U$ \\ \hline
    AGV       &$0.161$&$0.145$\\
    T-min \eqref{eq:sym-GUM-like}    &$0.174$&$0.122$\\
    U-min \eqref{modifiedsymGUM}    &$0.265$&$0.032$
        \end{tabular}
    \end{center}
    At first it may seem contradictory that AGV gives a lower transfer variance than the transfer variance minimizing GUE-implementation, but this is simply explained by AGV not implementing a GUE of truth telling here. Once again, the utility variance minimizer rule gives a dramatically lower utility variance at the expense of higher transfer variance.
    
\end{Example}\bigskip

Both results are conditional on the fixed decision policy. They compare transfer rules paired with that policy; they do not optimize jointly over all efficient decision policies.

Throughout the paper, private reports within each period are simultaneous. The variance minimizing transfer rules are not claimed to be GUE-implementations when sequential within-period reporting occurs.

\section{The Setup}

There is a set $\mathcal{I}=\{1,...,N\}$ of agents. Time is discrete, and the number of periods is finite. The time periods are indexed by $t \in \{ 0,1, ... , T \}$. In each period $t$, all the agents observe the (verifiable) public state $\theta_t^0 \in \boldsymbol{\Theta^0}$, and each agent $i \in \mathcal{I}$ privately observes his type $\theta^i_t \in \boldsymbol{\Theta^i}$. We denote the state space by $ \boldsymbol{\Theta} = \prod^{N}_{i=0} \boldsymbol{\Theta^i}$. The initial state $ \theta_0 \in \boldsymbol{\Theta}$ is fixed. After the state $\theta_t \in \boldsymbol{\Theta}$ is realized, a public decision $x_t  \in \boldsymbol{X}$ is made, and each agent $i$ collects a transfer $y_t^i \in \mathbb{R}$. 
We assume that transfers are \textit{budget-balanced}, meaning that in every period $t$ the transfers sum to zero:
$\sum\limits_{i\in\mathcal I} y_t^i \equiv 0$. In this paper, the type and decision spaces are finite.

We now describe the transition probabilities and the distribution of subsequent states. The transition probability function is denoted by $\mu: \boldsymbol{\Theta} \times \boldsymbol{X}   \rightarrow \Delta (\boldsymbol{\Theta})$. That is, given any state $\theta_t \in \boldsymbol{\Theta}$ and a public decision $x_t \in \boldsymbol{X}$ at period $t$, the ($t+1$)-period state $\theta_{t+1}$ is a random variable distributed according to the probability measure $\mu(\theta_t, x_t) \in \Delta (\boldsymbol{\Theta})$. \footnote{See \citet{BVJEL19} for a discussion of Markovian transition functions.} 
In addition, we assume that \textit{types are independent}: the transition probability can be written as 
\begin{align}\label{eq:product-transition}
\mu(\theta, x) = \mu^0(\theta^0, x) \times \prod_{i=1}^{N} \mu^i(\theta^0,\theta^i, x),
\end{align}
where $\mu^0: \boldsymbol{\Theta^{0}} \times \boldsymbol{X} \rightarrow \Delta (\boldsymbol{\Theta^0})$ and $\mu^i:  \boldsymbol{\Theta^{0}} \times  \boldsymbol{\Theta^{i}} \times  \boldsymbol{X}\rightarrow \Delta (\boldsymbol{\Theta^i})$.\footnote{This Markov formulation is without loss of generality. It is well-known that any dynamic model can be described using Markov notation by defining large enough private states (see, e.g., \citet{AtSe} for a discussion). The independent-product restriction above is a substantive assumption.} Intuitively, it implies that conditional on public information, an agent's private type does not affect the distribution of the current and future types of other agents.\footnote{This, however, does not prevent one agent's reports from affecting future types of other agents, as well as the future public state through the public decisions.}

\paragraph{Utilities.} In this paper we only consider the transferable utility setting. We denote by $y^i_t \in \mathbb{R}$ a transfer to agent $i \in \mathcal{I}$ at period $t$. For given sequences of types $\theta^i =\{ \theta^i_t \}_{t=1}^{T}$, public decisions $ x =\{ x_t \}_{t=0}^{T}$, and realized transfers $y^i =\{ y^i_t \}_{t=0}^{T}$, the utility of agent $i \in \mathcal{I}$ is given by
  $U^i(\theta^i, x, y^i) = \sum_{t=0}^{T} \bigg[ u^i(\theta^i_t, x_t) + y^i_t \bigg]$, where $\te_0^i\in\Te^i$ has been previously fixed.
That is, we have an environment with \textit{private values} since the per-period payoff function $u^i(\theta^i_t, x_t)$ of each agent directly depends only on his private type and public decisions. Intuitively, this means that each agent can calculate his utility by only observing the public decision and his type, even if he is uncertain about other agents' types. Note that although the agents' utilities are quasilinear, we assume they can depend on private signals and public decisions in a general way.

\paragraph{Mechanism.}
A mechanism consists of the following.
\begin{enumerate}
    \item A pure \emph{decision policy} is a function $\mathsf{x}: \big(\{h_0\}\cup\boldsymbol{\Theta} \cup (\boldsymbol{\Theta})^2 \cup \dots\cup (\boldsymbol{\Theta})^{T}\big)\to \boldsymbol{X}$, that determines the public decision $x_t = \mathsf{x}(\h\te_1, \widehat{\theta}_2, \dots \widehat{\theta}_t)=\mathsf{x}(h_t)$ at every time period $t>0$ and also the initial decision $x_0=\x(h_0)$, where $h_0$ represents the empty history.
    A decision policy is a probability distribution over pure decision policies.   
    \item A \emph{transfer rule} $y_t^i = y_t^i(h_t)$ for every agent $i \in \mathcal{I}$, and every time period $t$, as a function of the public history.
\end{enumerate}

\paragraph{Timing.} At period $t=0$, all the agents observe the initial state $ \theta_0  = \widehat{\theta}_0 \in \boldsymbol{\Theta}$, and a decision $ x_0 \in \boldsymbol{X}$ is made.
Then, in each period $t \geqslant 1$, all the agents observe the verifiable public state $\theta_t^0 \in \boldsymbol{\Theta}^0$, and each agent $i \in \I$ privately observes his type $\theta_t^i \in \boldsymbol{\Theta}^i$.
After observing the public state and his private type, every agent $i \in \I$ makes a public announcement $\widehat{\theta}_t^i \in \boldsymbol{\Theta}^i$ of his private type. These announcements are made \emph{simultaneously}, so an agent $i\in\I$ does not know the period-$t$ report $\h\te_t^j$ of another agent $j\in\I\setminus\{i\}$ when making his period-$t$ report $\h\te_t^i$, unless $i$ and $j$ collude. Given the announced types and the public state, transfers $y_t^i$ to each agent are calculated, and a decision $x_t \in \boldsymbol{X}$ is made.

\paragraph{Strategy.} Let $\widehat{\theta}_{t} = \big( \widehat{\theta}_t^0,  \widehat{\theta}_t^1, ...,  \widehat{\theta}_{t}^{N} \big)$ denote the vector of reported types in period $t$ including the public signal $\widehat{\theta}_t^0 = \theta_t^0$.
The public history up to the period $t$ is $h_t = (\widehat{\theta}_1, \ldots, \widehat{\theta}_t)$,
and we denote by $\mathcal{H}_t$ the set of all such public histories.\footnote{We omit $(x_0,\ldots,x_t)$ from the history since, under the mechanisms we consider, public decisions are (possibly via publicly recorded randomness) functions of the recorded reports $(\widehat{\theta}_1,\ldots,\widehat{\theta}_t)$.}
Let $h_{t,\emptyset}=(\widehat{\theta}_1, \ldots, \widehat{\theta}_{t-1}, \widehat{\theta}_t^0)$ denote the partial public history after the period-$t$ public state has been observed and before the private reports in period $t$ are submitted, and let $\mathcal{H}_{t,\emptyset}$ denote the set of all such partial public histories.
A pure strategy of agent $i$ is given by $s^i = \{s^i_t\}_{t=1}^{T}$, where $s^i_t : \big( \boldsymbol{\Theta^{i}}\big)^{t} \times  \mathcal{H}_{t,\emptyset} \rightarrow  \boldsymbol{\Theta^{i}}$ maps agent $i$'s private types up to period $t$ and the partial public history before the private reports in period $t$ to a report in period $t$.
We denote the set of all mixed strategies for agent $i$ by $\S^{i}$.
We say that an agent follows a truth-telling strategy if $s^i_t \equiv \theta^i_t$ for every period $t \geqslant 1$. In period $0$, all agents are truth-telling, so $\h\te_0=\te_0$.

To accommodate various collusive schemes, we also introduce a collusive environment that allows agents to jointly coordinate their reported types. We assume that the reports are enforced by a (benevolent) coordinator who has all the information the agents in the coalition have.
Specifically to our model, a coordinated strategy of a set of agents $L \subseteq \I$ is determined by $\bar{s}^L = \{ \bar{s}^L_t \}_{t=1}^{T}$, where $\bar{s}^L_t: (\boldsymbol{\Theta}^L)^{t} \times \mathcal{H}_{t,\emptyset} \rightarrow \boldsymbol{\Theta}^L$.
A pure \emph{coordinated strategy} $\bar{s}^L$ maps the joint information of the agents in $L$ and the partial public history before the private reports in period $t$ to reports $\widehat{\theta}^L$.
We denote by $\bar{\S}^L$ the set of probability distributions on these pure strategies.

Coalition-internal side payments may also be allowed. They do not affect the utility of an outsider who follows truth-telling, nor the coalition's aggregate utility, and therefore do not alter the arguments below.

\textbf{Unless otherwise stated, proofs are written for pure strategies.} The corresponding expected-utility statements extend to mixed strategies by averaging over pure realizations when the expectations are ex ante. We handle other cases with special care. \footnote{From this point onward, the hat placed over $\te\in(\Te)^T$ to indicate that it is a report profile will be treated as a function symbol. For example, given a fixed pure strategy profile $s\in{\S}^{\I}$, we define the mapping $\te\mapsto\widehat{\te}\in(\Te)^T$ recursively as follows: for every $t>0$, $\widehat{\te}_t^0=\te_t^0$ and
$\widehat{\te}_t^i=s_t^i(\te_{\leq t}^i,h_{t,\emptyset})$
for every $i\in\I$, where $\te_0=\h\te_0$ has been previously fixed. In doing so, we implicitly assume that $s$ is deterministic.}

\subsection{Partial histories and continuation values}\label{notation}
If $n$ is a positive integer, then let $[n]=\{1,2,\dots,n\}$, while let $[0]=\emptyset$.

For $H\subseteq\I$ and $t\in[T]$, let
\begin{equation}\label{eq:partial-history}
h_{t,H}
=
\bigl(\widehat\theta_1,\ldots,\widehat\theta_{t-1},
\widehat\theta_t^0,(\widehat\theta_t^i)_{i\in H}\bigr),
\end{equation}
where $(\widehat\theta_t^i)_{i\in H}$ is an $H$-indexed family, and let $h_{0,\I}$ equal the empty history $h_0$. For an element of the state space $\theta^*\in(\Theta)^{T}$, we analogously define $h_{t,H}(\theta^*)=(\te^*_1\dots,\te^*_{t-1},{\te^*}_t^0,({\te^*}_t^i)_{i\in H})$ and $h_t(\te^*)=(\te^*_1,\dots,\te^*_t)$, and let $h_{0,\I}(\te^*)=h_0(\te^*)=h_0$. 

We denote the \textbf{anticipated distribution of states} given a report history $h_{t,H}$ and a decision policy $\mathsf{x}$ by $\mu(h_{t,H},\mathsf{x})$. $\mu(h_{t,H},\mathsf{x})$ is the continuation law obtained by fixing the coordinates recorded in $h_{t,H}$ and the initial reports $\h\te_0$, and recursively drawing all unrecorded current and future coordinates from the transition kernels in \eqref{eq:product-transition}, applying $\mathsf x$ to the resulting histories.\footnote{This definition is meaningful even for a report history that has zero probability under truthful reporting.}

Given a state history $h_{t,H}(\theta)$, a (pure) decision policy $\mathsf x$, and a (pure) grand coalition coordinated strategy $\overline s\in\overline{\S}^{\I}$, write $\mu(h_{t,H}(\theta),\mathsf x,\overline s)$ for the actual continuation law of states when the coordinates in $h_{t,H}(\te)$ and the initial state $\te_0$ are fixed, reports are generated by $\overline s$, and decisions are made according to $\x$. Under truth-telling this agrees with $\mu(h_{t,H}(\theta),\mathsf x)$.

\begin{Def}[Continuation law expectation operators] Given a period $t>0$ and set of agents $H\subseteq\I$ (or the pair $t=0$ and $H=\I$), define the history kernel
\begin{align}
    \eta_{t,H}=(\{1,\dots,t-1\}\times(\{0\}\cup\I))\cup(\{t\}\times(\{0\}\cup H)).
\end{align}
Given a (pure) coordinated strategy $\overline{s}\in\overline{\S}^{\I}$ and a (pure) decision policy $\mathsf{x}$, we define the operator $\cE_{\mu(\eta_{t,H},\mathsf{x},\overline{s})}\lt[\cdot\rt]$:
    \begin{align}
        \cE_{\mu(\eta_{t,H},\mathsf{x},\overline{s})}\lt[\cdot\rt]:((\Te)^{T}\to\mathbb{R})&\to((\Te)^{T}\to\mathbb{R}),\\
        f&\mapsto\cE_{\mu(\eta_{t,H},\mathsf{x},\overline{s})}\lt[f\rt],\\
        \cE_{\mu(\eta_{t,H},\mathsf{x},\overline{s})}\lt[f\rt](\te^*)&\coloneq\E_{\mu(h_{t,H}(\te^*),\mathsf{x},\overline{s})}[f(\td\te)].
    \end{align}
    When $\overline{s}$ is the profile of truth-telling strategies, we suppress it from the notation.
\end{Def}
$\mu(\eta_{t,H},\mathsf{x},\overline{s})$ is not a probability measure itself, just a probability kernel.
\begin{Lemma}[A kind of linearity; the law of total expectation]\label{lem:LOTE}
$f$ and $g$ denote two $(\Te)^{T}\to\mathbb{R}$ functions, and the decision policy $\x$ and coordinated strategy $\overline{s}\in\overline{S}^{\I}$ are pure.
    \begin{enumerate}[label=(\roman*)]
        \item If $g(\te^*)=g(\te^{**})$ whenever $h_{t,H}(\te^*)=h_{t,H}(\te^{**})$, then $\cE_{\mu(\eta_{t,H},\mathsf{x},\overline{s})}\lt[fg\rt]=\cE_{\mu(\eta_{t,H},\mathsf{x},\overline{s})}\lt[f\rt]g$.
        \item $\cE_{\mu(\eta_{t_1,H_1},\mathsf{x},\overline{s})}\lt[\cE_{\mu(\eta_{t_2,H_2},\mathsf{x},\overline{s})}\lt[f\rt]\rt]=\cE_{\mu(\eta_{t_1,H_1}\cap\eta_{t_2,H_2},\mathsf{x},\overline{s})}\lt[f\rt]$.
    \end{enumerate}
\end{Lemma}
\begin{proof}
    See appendix \ref{app:LOTE}.
\end{proof}

For agent $j$, define the anticipated total payoff
\begin{equation}\label{eq:upsilon}
\Upsilon_{t,H,\mathsf x}^j
=
\cE_{\mu(\eta_{t,H},\mathsf x)}
\left[
\sum_{\tau=0}^T
u^j\bigl(\widetilde\theta_\tau^j,
\mathsf x(h_\tau(\widetilde\theta))\bigr)
\right],
\end{equation}
For $K\subseteq\I$, let
\[
\Upsilon_{t,H,\mathsf{x}}^K=\sum_{j\in K}\Upsilon_{t,H,\mathsf{x}}^j,
\qquad
\Upsilon_{t,H,\mathsf{x}}^{\emptyset}=0.
\]

\begin{Note}\label{note:total payoff}
    Given truth-telling of all agents, the total payoff of agent $i$ is $\Upsilon_{T,\I,\mathsf{x}}^i(\te)$, where $\te\in(\Te)^{T}$ is the state profile. More generally, if agent $i$ is truth-telling, then his total payoff is $\Upsilon_{T,\I,\mathsf{x}}^i(\h\te)$, where $\h\te\in(\Te)^{T}$ is the report profile.
\end{Note}

For each agent, let
\[
Y^i=\sum_{t=0}^T y_t^i
\]
denote the total transfer.

For a (pure) grand coalition coordinated strategy $\overline s\in\overline{S}^{\I}$, a (pure) decision policy $\mathsf{x}$, a transfer rule $y^i$ for agent $i$ and a state profile $\te\in(\Te)^{T}$, we denote the utility of agent $i$ given by these parameters by $U^i(\mathsf{x},\overline s,y^i)(\te)=\sum_{t=0}^T\lt[u^i(\te_t^i,\mathsf{x}(h_t(\h\te)))+y_t^i(\h\te)\rt]$, where the initial state $\te_0$ has already been fixed.\footnote{The transfer $y_t^i$ to agent $i$ in period $t$ is a function of the public history $h_t=h_t(\h\te)$ up to and including period $t$. Nevertheless, we will write it as function of the full report profile $\h\te$ as this is algebraically more convenient, but the fact that this notational convention is used does not mean that future reports can influence current transfers.} Here, $U^i(\mathsf{x}, \overline s,y^i)$ itself is a $(\Te)^{T}\to\mathbb{R}$ function. If $\overline s$ is the profile of truth-telling strategies, we suppress it from the notation, writing just $U^i(\mathsf{x},y^i)$.

\section{Fixing the decision policy, defining GUE-implementation of truthful reporting}

\paragraph{The decision policy.}In this paper, we will consider the decision policy to be fixed. We fix the decision policy such that it maximizes expected total payoff under truthful reporting.
\begin{equation}\label{thedecisionpolicy}
    \mathsf{x}\in \underset{\mathsf{x}'}{\arg\max}\left(\E_{\mu(h_{0,\I},\mathsf{x}')}\left[\sum_{i\in\I}\sum_{t=0}^Tu^i(\td\te_t^i,\mathsf{x}'(h_t(\td\te)))\right]\right).
\end{equation}
Since all relevant sets are finite, an optimal decision policy exists.

We fix the decision policy to be pure. We weakly justify this choice by showing that an efficient pure decision policy always exists. As stated earlier, a (mixed) decision policy is a probability distribution over pure decision policies, and its expected total payoff is the corresponding convex combination of the expected total payoffs generated by those pure policies. Consequently, if a randomized decision policy is optimal, then every pure decision policy in its support is also optimal. So a deterministic efficient decision policy $\mathsf{x}$ exists.

\begin{Lemma}[Truthful reporting maximizes total utility]\label{utilitysum}
    When the decision policy is fixed this way \eqref{thedecisionpolicy}, then for any coordinated strategy $\bar s\in\bar{\S}^\I$,
    \begin{align}
        \sum_{i\in\I}\E_{\mu(h_{0,\I},\mathsf{x},\bar s)}\lt[U^i(\mathsf{x},\overline{s},y^i)\rt]\leq\sum_{i\in\I}\E_{\mu(h_{0,\I},\mathsf{x})}\lt[U^i(\mathsf{x},y^i)\rt]
    \end{align}
    for any budget-balanced transfer rule, where the transfer to agent $i$ at period $t$ is denoted by $y_t^i(h_t)=y_t^i(\h\te)$, and $y^i=\{y_t^i\}_{t=0}^T$.
\end{Lemma}
\begin{proof}See Appendix \ref{app:utilitysum}.\end{proof}

\paragraph{}\label{simplifiednotation}From now on, $\mathsf x$ is fixed and suppressed from the notation. Thus $\mu(h_{t,H})$ means $\mu(h_{t,H},\mathsf x)$, $\Upsilon_{t,H}^j$ means $\Upsilon_{t,H,\mathsf{x}}^j$, et cetera.

\paragraph{GUE-implementations.}In this paper, we will refer to transfer rules as GUE-implementations of truthful reporting, meaning that when the fixed efficient decision policy is paired with the transfer rule, the resulting mechanism is a GUE-implementation of truthful reporting. In light of this, the definition of a GUE-implementation of truthful reporting takes the following form.

\begin{Def}[GUE-implementation of truthful reporting]\label{def:GUE-transfer}
   The (budget-balanced) transfer rule $y=\{y^i\}_{i\in\I}$, $y^i=\{y_t^i\}_{t=0}^T$, $y_t^i(h_t)=y_t^i(\h\te)$ is a GUE-implementation of truthful reporting, if for every agent $i\in\I$ and coordinated strategy of the other agents $\bar s^{-i}\in\bar\S^{-i}$,
   \begin{align}\label{GUE property}
   \E_{\mu(h_{0,\I},\bar s^{-i},s_*^i)}\lt[U^i(\overline{s}^{-i},s_*^i,y^i)\rt]\geq\E_{\mu(h_{0,\I})}\lt[U^i(y^i)\rt],
   \end{align}
   where $s_*$ is the profile of truth-telling strategies.
\end{Def}

\begin{Proposition}[Equivalence with GUE]\label{prop:equivalence with GUE}
    With the pure decision policy $\mathsf{x}$ fixed according to \eqref{thedecisionpolicy}, a budget-balanced transfer rule $y$ satisfies \eqref{GUE property} if and only if the mechanism $(y,\mathsf{x})$ is a GUE-implementation of truthful reporting.
\end{Proposition}
\begin{proof}

By the fifth transferable-utility GUE characterization of \citet[Lemma~2.3]{GUE}, all we have to show is that

\begin{gather}
    \forall\overline{s}^{-i}\in\overline{\S}^{-i}:\E_{\mu(h_{0,\I},\overline{s}^{-i},s_*^i)}\lt[U^i(\overline{s}^{-i},s_*^i,y^i)\rt]\geq\E_{\mu(h_{0,\I})}\lt[U^i(y^i)\rt]\\ \Updownarrow\\
    \begin{cases}
        \forall\overline{s}^{-i}\in\overline{\S}^{-i}:\E_{\mu(h_{0,\I},\overline{s}^{-i},s_*^i)}\lt[U^i(\overline{s}^{-i},s_*^i,y^i)\rt]\geq\E_{\mu(h_{0,\I})}\lt[U^i(y^i)\rt],\\
        \forall\overline{s}\in\overline{\S}^\I:\sum_{i\in\I}\E_{\mu(h_{0,\I},\overline{s})}\lt[U^i(\overline{s},y^i)\rt]\leq\sum_{i\in\I}\E_{\mu(h_{0,\I})}\lt[U^i(y^i)\rt]
    \end{cases}.
\end{gather}
Taking into account Lemma \ref{utilitysum}, the proof is complete.
\end{proof}

\section{Conditional characterization of GUE-implementation of truthful reporting}

In this section, we describe a sufficient and a necessary condition for GUE-implementation of truthful reporting, both of which use the same conditional equalities to characterize GUE. Together, the two conditions do not generally give an equivalent characterization of GUE-implementation of truthful reporting, but they do, for instance, when every state profile occurs with positive probability under truth-telling.

\subsection{Necessary condition}

\begin{Lemma}[Conditional guarantee]\label{Proposition 1}
    If the transfer rule $y$ is a GUE-implementation of truthful reporting, then for every agent $i\in\I$ and period $t\in[T]$,
    \begin{align}\label{eq:conditional guarantee}
        \cE_{\mu(\eta_{t,\emptyset},\overline{s}^{-i},s_*^i)}\lt[U^i(\overline{s}^{-i},s_*^i,y^i)\rt]\geq\cE_{\mu(\eta_{t,\emptyset})}\lt[U^i(y^i)\rt]\qquad\mu(h_{0,\I})\text{-a.s.}
    \end{align}
    if the pure coordinated strategy $\overline{s}^{-i}$ of the other agents is truth-telling up to and including period $t-1$, where $s_*$ is the profile of truth-telling strategies.
\end{Lemma}
\begin{proof}
    Let $\Psi\subseteq(\Te)^{T}$ denote the set of state paths where the inequality \eqref{eq:conditional guarantee} fails. Our goal is to show that this is a zero probability event under $\mu(h_{0,\I})$. Note that whether a state profile $\te\in(\Te)^{T}$ falls within $\Psi$ is already determined by the partial history $h_{t,\emptyset}(\te)$, as both sides of the inequality \eqref{eq:conditional guarantee} only depend on this partial history. It is then very natural to introduce the set of partial histories $\Psi_{t,\emptyset}=\{h_{t,\emptyset}(\te)\mid\te\in\Psi\}$ where the inequality \eqref{eq:conditional guarantee} fails. $\bar s\coloneqq(\bar s^{-i},s_*^i)$.

    We construct the coordinated strategy $\overline{s}'\in\overline{\S}^\I$ where agents report according to $\overline{s}$, except when they perceive a reported partial history $h_{t,\emptyset}(\te)\notin\Psi_{t,\emptyset}$, in which case they continue truth-telling. Because $y$ is a GUE implementation of truthful reporting,
    \begin{align}\label{eq:prior inequality}
        \E_{\mu(h_{0,\I},\overline s')}\lt[U^i(\overline{s}',y^i)\rt]\geq\E_{\mu(h_{0,\I})}\lt[U^i(y^i)\rt].
    \end{align}
    By the tower property, the right side of \eqref{eq:prior inequality} can be written as
    \begin{align}\label{right1}
        \E_{\mu(h_{0,\I})}\lt[U^i(y^i)\rt]=\E_{\mu(h_{0,\I})}\lt[\cE_{\mu(\eta_{t,\emptyset})}\lt[U^i(y^i)\rt]\rt],
    \end{align}
    and the left
    \begin{align}\label{right2}        &\E_{\mu(h_{0,\I},\overline s')}\lt[U^i(\overline{s}',y^i)\rt]=\E_{\mu(h_{0,\I},\overline s')}\lt[\cE_{\mu(\eta_{t,\emptyset},\overline s')}\lt[U^i(\overline{s}',y^i)\rt]\rt]\\&=\E_{\mu(h_{0,\I})}\lt[\cE_{\mu(\eta_{t,\emptyset},\overline s')}\lt[U^i(\overline{s}',y^i)\rt]\rt],
    \end{align}
    additionally considering that $\overline s'$ is truth-telling before period $t$. Because $\overline{s}'$ is truth-telling conditional on a reported partial history outside of $\Psi_{t,\emptyset}$ having occurred, $\cE_{\mu(\eta_{t,\emptyset},\overline{s}')}\lt[U^i(\overline{s}',y^i)\rt](\te^*)=\cE_{\mu(\eta_{t,\emptyset})}[U^i(y^i)](\te^*)$ when $\te^*\notin\Psi$. Rearranging \eqref{eq:prior inequality} gives
    \begin{align}\label{eq:rearranged prior inequality}
        \E_{\mu(h_{0,\I})}\lt[\cE_{\mu(\eta_{t,\emptyset})}\lt[U^i(y^i)\rt]-\cE_{\mu(\eta_{t,\emptyset},\overline s')}\lt[U^i(\overline{s}',y^i)\rt]\rt]\leq 0.
    \end{align}
    Here, the integrand is zero on state paths outside of $\Psi$ and positive on state paths within $\Psi$; all the while the integral is non-positive. This implies that $\Psi$ has zero probability under $\mu(h_{0,\I})$.
\end{proof}

This lemma naturally extends to mixed strategies by taking the corresponding convex combination of inequalities of the form \eqref{eq:conditional guarantee}.

\begin{Proposition}[Necessary on-path condition]\label{cor1}
If the transfer rule $y$ is a GUE-implementation of truthful reporting, then for every period $t\in[T]$ and agent $i\in\I$
\begin{equation}
    \cE_{\mu(\eta_{t,\I\setminus\{i\}})}\lt[U^i(y^i)\rt]=\cE_{\mu(\eta_{t,\emptyset})}\lt[U^i(y^i)\rt]\qquad\mu(h_{0,\I})\text{-a.s.}
\end{equation}
\end{Proposition}
Of course, if every state profile occurs with positive probability under truth-telling, then "$\mu(h_{0,\I})$-a.s." can be replaced by "pointwise".

\begin{proof}[Proof of Proposition \ref{cor1}]
    Fix an agent $i$ and a state path $\te\in(\Te)^T$ that occurs with positive probability under truth-telling. Write $C^i$ for the utility guarantee of agent $i$ conditional on the state history $h_{t,\emptyset}(\te)$ if he reports truthfully and all other agents report truthfully before period $t$:
    \begin{align}
        C^i=\inf\{\E_{\mu(h_{t,\emptyset}(\te),\overline{s})}\lt[U^i(\overline{s},y^i)\rt]\mid s\in\overline{\S}^\I\text{ is a pure coordinated strategy, where }\cr s^i\text{ is truth-telling, and }s^j\text{ is truth-telling before period }t\text{ for each }j\neq i\}.
    \end{align}
    Write $D^i$ for the utility guarantee of agent $i$ conditional on the state history $h_{t,\emptyset}(\te)$ if he reports truthfully, all other agents are truth-telling before period $t$, and in period $t$ each agent $j\neq i$ reports $\te_t^j$:
    \begin{align}
        D^i=\inf\{\E_{\mu(h_{t,\emptyset}(\te),\overline{s})}\lt[U^i(\overline{s},y^i)\rt]\mid s\in\overline{\S}^\I\text{ is a pure coordinated strategy, where }\cr s^i\text{ is truth-telling, and }s^j\text{ is truth-telling before period }t\text{ and reports }\cr \te_t^j\text{ in period }t\text{ for each }j\neq i\}.
    \end{align}
    Lemma \ref{Proposition 1} states that $C^i\geq\E_{\mu(h_{t,\emptyset}(\te))}\lt[U^i(y^i)\rt]$. $D^i\geq C^i$, because shrinking a set does not decrease its infimum. Lastly, $\E_{\mu(h_{t,\I\setminus\{i\}}(\te))}\lt[U^i(y^i)\rt]\geq D^i$, because $\E_{\mu(h_{t,\I\setminus\{i\}}(\te))}\lt[U^i(y^i)\rt]$ is a convex combination of elements of the set whose infimum is $D^i$: after each agent $j\neq i$ reports $\te_t^j$ in period $t$, in the subsequent periods they distribute their reports according to the anticipated distribution of their states. By transitivity
    \begin{align}
        E_{\mu(h_{t,\I\setminus\{i\}}(\te))}\lt[U^i(y^i)\rt]\geq\E_{\mu(h_{t,\emptyset}(\te))}\lt[U^i(y^i)\rt].
    \end{align}
    
    This holds for every state profile with positive probability under truth-telling, meaning
    \begin{align}\label{eq:a.s. inequality 1}
        \cE_{\mu(\eta_{t,\I\setminus\{i\}})}\lt[U^i(y^i)\rt]\geq\cE_{\mu(\eta_{t,\emptyset})}\lt[U^i(y^i)\rt]\qquad\mu(h_{0,\I})\text{-a.s.},
    \end{align}
    i.e.
    \begin{align}\label{eq:a.s. inequality}
        \cE_{\mu(\eta_{t,\I\setminus\{i\}})}\lt[U^i(y^i)\rt]-\cE_{\mu(\eta_{t,\emptyset})}\lt[U^i(y^i)\rt]\geq0\qquad\mu(h_{0,\I})\text{-a.s.}
    \end{align}
    Furthermore, by the tower property,
    \begin{align}
        \E_{\mu(h_{0,\I})}\lt[\cE_{\mu(\eta_{t,\I\setminus\{i\}})}\lt[U^i(y^i)\rt]-\cE_{\mu(\eta_{t,\emptyset})}\lt[U^i(y^i)\rt]\rt]=0,
    \end{align}
    hence the integrand, which we already know to be non-negative almost surely, is zero almost surely. The proof is complete.
\end{proof}

\subsection{Sufficient condition}
The next result gives a sufficient off-path condition. It is stronger than Corollary~\ref{cor1}, because arbitrary deviations may generate report histories that have zero probability under truthful reporting.

\begin{Lemma}[Two conditional-law identities]\label{lem:conditional-laws}\footnote{This lemma is quite setup-specific. It requires that announcements are made simultaneously in each period and the independent types assumption.}
Fix a period $t>0$ and an agent $i$ who reports truthfully, while the other agents may use an arbitrary pure coordinated strategy $\overline{s}^{-i}$. Fix a state profile $\te\in(\Te)^{T}$, and for an arbitrary $\te^*\in(\Te)^{T}$ write $\h{\te^*}$ for the reports generated from $\te^*$ by the coordinated strategy $\overline{s}=(\overline{s}^{-i},s_*^i)$, where $s_*^i$ is $i$'s truth-telling strategy. $f$ and $g$ denote functions of the displayed histories.

\begin{enumerate}[label=(\roman*)]
\item $\cE_{\mu(\eta_{t,\I\setminus\{i\}},\overline{s})}\lt[f(h_{t,\I}(\h{\td\te}))\rt](\te)=\cE_{\mu(\eta_{t,\I\setminus\{i\}})}\lt[f(h_{t,\I}(\td\te))\rt](\h\te)$.
\item $\cE_{\mu(\eta_{t-1,\I},\overline{s})}\lt[g(h_{t,\emptyset}(\h{\td\te}))\rt](\te)=\cE_{\mu(\eta_{t-1,\I})}\lt[g(h_{t,\emptyset}(\td\te))\rt](\h\te)$.\footnote{Part (ii) does not actually require any agent to be truthful.}
\end{enumerate}
\end{Lemma}
\begin{proof}
    \begin{enumerate}[label=(\roman*)]
        \item By the definition of the transition probability function \eqref{eq:product-transition}, the distribution of $i$'s period-$t$ state is the same under the continuation law $\mu(h_{t,\I\setminus\{i\}}(\te),\overline{s})$ and $\mu(h_{t,\I\setminus\{i\}}(\h\te))$, as $i$ is truth-telling and the public state is verifiable. On the left, because the other colluding agents cannot observe $i$'s current true type or $i$'s current report or the future when formulating their reports, and because $i$ is truth-telling,
        \begin{align}
            h_{t,\I}(\h{\td\te})=(h_{t,\I\setminus\{i\}}(\h\te),\td\te_t^i),\qquad\text{where }\td\te\sim\mu(h_{t,\I\setminus\{i\}}(\te),\overline{s}).
        \end{align}
        On the right,
        \begin{align}
            h_{t,\I}(\td\te)=(h_{t,\I\setminus\{i\}}(\h\te),\td\te_t^i),\qquad\text{where }\td\te\sim\mu(h_{t,\I\setminus\{i\}}(\h\te)).
        \end{align}
        So on both sides, $f$ receives the same input, and the two displayed histories have the same distribution. Thus, the two sides match.
        \item Similarly to (i), the distribution of the period-$t$ public state is the same under the continuation laws $\mu(h_{t-1,\I}(\te),\overline{s})$ and $\mu(h_{t-1,\I}(\h\te))$ as the public state is verifiable. On the left,
        \begin{align}
            h_{t,\emptyset}(\h{\td\te})=(h_{t-1,\I}(\h\te),\td\te_t^0),\qquad\text{where }\td\te\sim\mu(h_{t-1,\I}(\te),\overline{s}),
        \end{align}
        because agents cannot observe the future, and because the public state is verifiable. On the right,
        \begin{align}
            h_{t,\emptyset}(\td\te)=(h_{t-1,\I}(\h\te),\td\te_t^0),\qquad\text{where }\td\te\sim\mu(h_{t-1,\I}(\h\te)).
        \end{align}
        The proof concludes as in (i).
    \end{enumerate}
\end{proof}

\begin{Proposition}[Sufficient off-path condition]\label{elegseges} 
    If $y$ is a budget-balanced transfer rule, and for every period $t\in[T]$ and agent $i\in\I$,
    \begin{align}\label{eq:elegseges}
        \cE_{\mu(\eta_{t,\I\setminus\{i\}})}\lt[U^i(y^i)\rt]=\cE_{\mu(\eta_{t,\emptyset})}\lt[U^i(y^i)\rt]\qquad\text{(pointwise)},
    \end{align}
    then $y$ is a GUE-implementation of truthful reporting.
\end{Proposition}
\begin{proof}
        Suppose the hypothesis. Fix an agent $i$ and a collusive strategy of the other agents $\overline s^{-i}\in\overline\S^{-i}$. Let $s_*$ denote the profile of truth-telling strategies.

        Because $i$ is truth-telling and by the private values assumption,
        \begin{align}
            \E_{\mu(h_{0,\I},\overline{s}^{-i},s_*^i)}\lt[U^i(\overline{s}^{-i},s_*^i,y^i)(\td\te)\rt]=\E_{\mu(h_{0,\I},\overline{s}^{-i},s_*^i)}\lt[U^i(y^i)(\h{\td\te})\rt]=\E_{\mu(h_{0,\I},\overline{s}^{-i},s_*^i)}\lt[\cE_{\mu(\eta_{T,\I})}\lt[U^i(y^i)\rt](\h{\td\te})\rt].
        \end{align}

        We will now show that for every $t>0$,
        \begin{align}
            \E_{\mu(h_{0,\I},\overline{s}^{-i},s_*^i)}\lt[\cE_{\mu(\eta_{t,\I})}\lt[U^i(y^i)\rt](\h{\td\te})\rt]=\E_{\mu(h_{0,\I},\overline{s}^{-i},s_*^i)}\lt[\cE_{\mu(\eta_{t-1,\I})}\lt[U^i(y^i)\rt](\h{\td\te})\rt].
        \end{align}
        For every $\te\in(\Te)^{T}$, Lemma \ref{lem:conditional-laws}(i) and the tower property give
        \begin{align}
            \cE_{\mu(\eta_{t,\I\setminus\{i\}},\overline{s}^{-i},s_*^i)}\lt[\cE_{\mu(\eta_{t,\I})}\lt[U^i(y^i)\rt](\h{\td\te})\rt](\te)=\cE_{\mu(\eta_{t,\I\setminus\{i\}})}\lt[U^i(y^i)\rt](\h\te).
        \end{align}
        By the hypothesis, this equals $\cE_{\mu(\eta_{t,\emptyset})}\lt[U^i(y^i)\rt](\h\te)$. For every $\te\in(\Te)^{T}$, Lemma \ref{lem:conditional-laws}(ii) and the tower property give
        \begin{align}
            \cE_{\mu(\eta_{t-1,\I},\overline{s}^{-i},s_*^i)}\lt[\cE_{\mu(\eta_{t,\emptyset})}\lt[U^i(y^i)\rt](\h{\td\te})\rt](\te)=\cE_{\mu(\eta_{t-1,\I})}\lt[U^i(y^i)\rt](\h\te).
        \end{align}
        Applying the law of total expectation twice, the above equalities yield
        \begin{align}
            &\E_{\mu(h_{0,\I},\overline{s}^{-i},s_*^i)}\lt[\cE_{\mu(\eta_{t,\I})}\lt[U^i(y^i)\rt](\h{\td\te})\rt]=\E_{\mu(h_{0,\I},\overline{s}^{-i},s_*^i)}\lt[\cE_{\mu(\eta_{t,\I\setminus\{i\}})}\lt[U^i(y^i)\rt](\h{\td\te})\rt]\\
            &=\E_{\mu(h_{0,\I},\overline{s}^{-i},s_*^i)}\lt[\cE_{\mu(\eta_{t,\emptyset})}\lt[U^i(y^i)\rt](\h{\td\te})\rt]=\E_{\mu(h_{0,\I},\overline{s}^{-i},s_*^i)}\lt[\cE_{\mu(\eta_{t-1,\I})}\lt[U^i(y^i)\rt](\h{\td\te})\rt].
        \end{align}

        So,
        \begin{align}
            \E_{\mu(h_{0,\I},\overline{s}^{-i},s_*^i)}\lt[\cE_{\mu(\eta_{T,\I})}\lt[U^i(y^i)\rt](\h{\td\te})\rt]=\dots=\E_{\mu(h_{0,\I},\overline{s}^{-i},s_*^i)}\lt[\cE_{\mu(\eta_{0,\I})}\lt[U^i(y^i)\rt](\h{\td\te})\rt],
        \end{align}
        which concludes the proof as $\cE_{\mu(\eta_{0,\I})}\lt[U^i(y^i)\rt](\h{\td\te})\equiv\E_{\mu(h_{0,\I})}[U^i(y^i)]$.
\end{proof}

Together, Proposition~\ref{cor1} and Proposition~\ref{elegseges} give an almost-equivalent conditional characterization of GUE-implementation of truthful reporting that becomes a true characterization under a full support assumption.

\section{Utility Variance Minimizing GUE-Implementation of truthful reporting}
In this section, we prove that a modified version of sym-GUM transfers minimizes the total utility variance when all agents are truth-telling.

\paragraph{GUM.}Fix the order $0,1,\dots,N$ of agents, where $0$ denotes the public signal. For distinct agents $i$ and $j$,
\begin{align}
    \gamma_t^{i\to j}=\Upsilon_{t,[i]}^j-\Upsilon_{t,[i-1]}^j.
\end{align}
The \textbf{GUM} transfer to agent $i$ is
\begin{align}\label{eq:GUM transfers}
    y_t^i=\sum_{j\neq i}\lt(\gamma_t^{i\to j}-\gamma_t^{j\to i}\rt),\qquad t>0,
\end{align}
with $y_0^i=0$. The terms $\Upsilon_{[j]}^i-\Upsilon_{[j-1]}^i$ telescope, giving the alternate form
\begin{align}\label{eq:GUM alterante formula}
    y_t^i=\lt(\Upsilon_{t,\emptyset}^i-\Upsilon_{t,\I}^i\rt)+\lt(\Upsilon_{t,[i]}^\I-\Upsilon_{t,[i-1]}^\I\rt),\qquad t>0.
\end{align}

GUM is analogously defined for an arbitrary order of agents.

\begin{Note}\label{GUMisGUE}
    The GUM transfer rule \eqref{eq:GUM transfers} is indeed a GUE-implementation of truthful reporting.
\end{Note}
\begin{proof}
    The transfer rule is budget-balanced as it is composed of bilateral transfers. We then verify \eqref{eq:elegseges} with the tower property, so by proposition \ref{elegseges}, the transfer rule is a GUE-implementation of truthful reporting. Appendix \ref{app:GUMisGUEcalculation} records the calculation in expanded form.
\end{proof}

\paragraph{Sym-GUM.}Sym-GUM is the transfer rule we obtain by calculating GUM transfers for all $N!$ orders of agents, and taking their average. Calculating the formula for \textbf{sym-GUM}, we get that the transfer to agent $i$ in period $t>0$ is
\begin{equation}\label{sym-GUM formula}
    \overline y_t^i=\Upsilon_{t,\emptyset}^i-\Upsilon_{t,\I}^i+\sum_{i\in H\subseteq\I}\frac{\Upsilon_{t,H}^\I-\Upsilon_{t,H\setminus\{i\}}^\I}{|H|\binom{N}{|H|}},
\end{equation}
and $\overline{y}_0^i=0$. See appendix \ref{GUM to sym-GUM} for the derivation of this. This transfer rule is a GUE-implementation of truthful reporting since it is a convex combination of GUE-implementations of truthful reporting.

It is quite intuitive that without the symmetrization step, GUM cannot minimize total utility variance. It is also somewhat intuitive that sym-GUM still cannot minimize total utility variance, because it lets the public state help one agent more than another. We present the following example that demonstrates this.

\begin{Example}[Why sym-GUM alone need not minimize utility risk]\label{ex:public-state-risk}
Let $N=2$ and $T=1$. Each private type space is a singleton. In period $1$, the public state is ``agent 1 wins'' or ``agent 2 wins'', each with probability $1/2$, independently of everything. There are three decisions: ``agent 1 won'', ``agent 2 won'', and ``stalemate''. The fixed policy chooses stalemate at $t=0$ and follows the public state at $t=1$.

If "agent 1 won" is decided, primitive payoffs are $(10,-10)$; if "agent 2 won" is, they are $(-10,10)$; under stalemate they are $(0,0)$. Total payoff is always zero, so the policy is efficient. Private reports have no effect, hence all GUM externalities and sym-GUM transfers are zero. Each agent's utility is therefore $10$ or $-10$.

Alternatively, transfer to each agent the negative of that agent's realized primitive payoff in period $1$. The rule is budget-balanced, and because no private report can vary, is a GUE-implementation of truthful reporting. Both agents then receive utility zero with certainty, so this transfer rule yields a strictly lower total utility variance than sym-GUM.
\end{Example}

\paragraph{Modified Sym-GUM.}This shortcoming of sym-GUM can easily be rectified. All we have to do is introduce additional transfers in each period that even out the effect of the public state. After the public state is revealed in some period $t>0$, every agent $i$ transfers one $N$th of the change in his anticipated total payoff caused by the revelation of the public state (i.e. $\frac{\Upsilon_{t,\emptyset}^i-\Upsilon_{t-1,\I}^i}{N}$) to all other agents. For each agent $i$, this means an 
\textbf{additional transfer} to $i$ of
\begin{equation}\label{additional transfers}
    {y^+}_t^i=\frac{\Upsilon_{t,\emptyset}^{\I\setminus\{i\}}-\Upsilon_{t-1,\I}^{\I\setminus\{i\}}}{N}-(N-1)\frac{\Upsilon_{t,\emptyset}^{i}-\Upsilon_{t-1,\I}^{i}}{N}=\frac{\Upsilon_{t,\emptyset}^{\I}-\Upsilon_{t-1,\I}^{\I}}{N}-\Upsilon_{t,\emptyset}^{i}+\Upsilon_{t-1,\I}^{i}
\end{equation}
in period $t>0$, with ${y^+}_0^i=0$. The formula we then get for the \textbf{modified sym-GUM transfer} to agent $i$ in period $t>0$ is
\begin{equation}\label{modifiedsymGUM}
    {\overline{y}^+}_t^i=\overline y_t^i+{y^+}_t^i=\frac{\Upsilon_{t,\emptyset}^\I-\Upsilon_{t-1,\I}^\I}{N}+\Upsilon_{t-1,\I}^i-\Upsilon_{t,\I}^i+\sum_{i\in H\subseteq\I}\frac{\Upsilon_{t,H}^\I-\Upsilon_{t,H\setminus\{i\}}^\I}{|H|\binom{N}{|H|}},
\end{equation}
and the transfer in period $0$ is $0$ for every agent.
\begin{Proposition}\label{prop:modified sym-GUM is GUE}
    This modified version of sym-GUM is still a GUE-implementation of truthful reporting.
\end{Proposition}
\begin{proof}
    Since we already know that sym-GUM is a GUE-implementation of truthful reporting, it will suffice to show that the additional transfers $\{{y^+}_t^i\}_{t=0,i\in\I}^T$ in \eqref{additional transfers} are budget-balanced, and that for every agent $i$, if he follows his truth-telling strategy, then regardless of what coordinated strategy the other agents choose, his expected total additional transfer will be the same as when all agents are truth-telling. In other words, we need
    \begin{align}
        \forall t=0,\dots,T:\sum_{i\in\I}{y^+}_t^i&\equiv0,\label{balance}\\
        \forall i\in\I,\forall \bar s^{-i}\in\bar\S^{-i}:\E_{\mu(h_{0,\I},s_*^i,\bar s^{-i})}\lt[{Y^+}^i\rt]&=\E_{\mu(h_{0,\I})}\lt[{Y^+}^i\rt],\label{guarantee}
    \end{align}
    where $s_*\in\S^\I$ is the profile of truth-telling strategies, and ${Y^+}^i=\sum_{t=0}^T{y^+}_t^i$. \eqref{balance} is trivial, so we only need to prove \eqref{guarantee}.

    By the law of total expectation, the right side of \eqref{guarantee} is $0$, so it will be sufficient to show that for every agent $i$ and period $t$,
    \begin{align}
        \E_{\mu(h_{0,\I},\bar s)}\lt[\Upsilon_{t,\emptyset}^k(\widehat{\td\te})-\Upsilon_{t-1,\I}^k(\widehat{\td\te})\rt]=0
    \end{align}
    for every coordinated strategy of the form $\overline{s}=(\overline{s}^{-i},s_*^i)$ and every $k\in\I$, where $\overline{s}^{-i}\in\overline\S^{-i}$.

    Fix an agent $i$, a period $t$, a coordinated strategy $\overline s=(\overline{s}^{-i},s_*^i)$, and a $k\in\I$. For every $\te\in(\Te)^{T}$, Lemma \ref{lem:conditional-laws}(ii) gives
\begin{align}
    \Upsilon_{t-1,\I}^k(\widehat\te)=\cE_{\mu(\eta_{t-1,\I},\bar s)}\lt[\Upsilon_{t,\emptyset}^k(\h{\td\te})\rt](\te).
\end{align}
Applying the law of total expectation, the proof is complete.
\end{proof}

Under truth-telling, the tower property gives
\begin{equation}\label{eq:candidate-zero-mean}
\E_{\mu(h_{0,\I})}[{\overline{Y}^+}^i]=0
\qquad\text{for every }i.
\end{equation}

\begin{Lemma}\label{lem:Z}
    Let the transfer rules $y$ and $y'$ be two GUE-implementations of truthful reporting, define
    \begin{align}
        Z^i=Y'^i-Y^i\qquad\text{for every }i.
    \end{align}
    Then
    \begin{align}
        &\sum_{i\in\I}Z^i\equiv0\text{, and}\label{eq:Z-balance}\\
        &\cE_{\mu(\eta_{t,\I\setminus\{i\}})}\lt[Z^i\rt]=\cE_{\mu(\eta_{t,\emptyset})}\lt[Z^i\rt]\quad\mu(h_{0,\I})\text{-a.s.}\label{eq:Z-conditional}
    \end{align}
    for every agent $i$ and period $t$.
\end{Lemma}
\begin{proof}
    \eqref{eq:Z-balance} follows from budget-balance, and for \eqref{eq:Z-conditional} see Proposition \ref{cor1}.
\end{proof}

\begin{Theorem}[Utility variance minimization]\label{minutilityvariance}
    For the fixed decision policy, the modified sym-GUM transfer rule \eqref{modifiedsymGUM} minimizes total utility variance when every agent is truth-telling, among all GUE-implementations of truthful reporting. Here, we define total utility variance as $\sum_{i\in\I}\Var\lt(U^i(y^i)(\td\te)\rt)$, where $\td\te\sim\mu(h_{0,\I})$.
\end{Theorem}
\begin{proof}
    First normalize any competing GUE rule $y'$ by subtracting the constant $\E[Y'^i]$ from agent $i$'s period-$0$ transfer. Budget balance is preserved because $\sum_i\E[Y'^i]=0$; the GUE inequalities and all variances are unchanged. Hence we may assume
    \begin{equation}\label{eq:normalization}
    \E_{\mu(h_{0,\I})}[Y'^i]=\E_{\mu(h_{0,\I})}[{\overline{Y}^+}^i]=0
    \qquad\text{for every }i.
    \end{equation}
    Let
    \begin{align}
        Z^i=Y'^i-{\overline{Y}^+}^i.
    \end{align}
    Under truth-telling, the primitive payoff is the same under both transfer rules, so
    \begin{align}
        U^i(y'^i)=U^i({\overline{y}^+}^i)+Z^i
    \end{align}

    By \eqref{eq:normalization}, minimizing the sum of utility variances is equivalent to minimizing the sum of their second moments. It therefore suffices to prove the orthogonality identity
    \begin{equation}\label{eq:utility-orthogonality}
    \sum_{i\in\I}\E_{\mu(h_{0,\I})}[U^i({\overline{y}^+}^i)Z^i]=0.
    \end{equation}
    Afterwards, the Pythagorean theorem gives
    \begin{align}
        \sum_{i\in\I}\E_{\mu(h_{0,\I})}\lt[(U^i(y'^i))^2\rt]=\sum_{i\in\I}\E_{\mu(h_{0,\I})}\lt[(U^i({\overline{y}^+}^i))^2\rt]+\sum_{i\in\I}\E_{\mu(h_{0,\I})}\lt[(Z^i)^2\rt].
    \end{align}
    Telescoping the terms $\Upsilon_{t-1,\I}^i-\Upsilon_{t,\I}^i$ in \eqref{modifiedsymGUM} gives
    \begin{align}        U^i({\overline{y}^+}^i)=\Upsilon_{0,\I}^i+\sum_{t=1}^T\frac{\Upsilon_{t,\emptyset}^\I-\Upsilon_{t-1,\I}^\I}{N}+\sum_{t=1}^T\sum_{i\in H\subseteq\I}\frac{\Upsilon_{t,H}^\I-\Upsilon_{t,H\setminus\{i\}}^\I}{|H|\binom{N}{|H|}}.
    \end{align}
    The first term contributes zero to \eqref{eq:utility-orthogonality} because it is constant and $\E_{\mu(h_{0,\I})}[Z^i]=0$. The second term is common across agents; its contribution vanishes pointwise by \eqref{eq:Z-balance}. It remains to consider the final double sum. Fix a period $t$, write
    \begin{align}
        w_H=\frac{1}{|H|\binom{N}{|H|}},\qquad m^i=\cE_{\mu(\eta_{t,\emptyset})}\lt[Z^i\rt],
    \end{align}
    and suppress the period index on $\Upsilon$. The contribution of the positive terms is
\begin{align}
P&=
\sum_{i\in\I}\E_{\mu(h_{0,\I})}\lt[\sum_{i\in H\subseteq\I}w_H\Upsilon_H^\I Z^i\rt]=\sum_{\emptyset\neq H\subseteq\I}\E_{\mu(h_{0,\I})}\lt[\sum_{i\in H}w_H\Upsilon_H^\I Z^i\rt]\\
&=-\sum_{\emptyset\ne H\subseteq\I}w_H\sum_{j\notin H}\E_{\mu(h_{0,\I})}[\Upsilon_H^{\I}Z^j],
\end{align}
where the third equality uses \eqref{eq:Z-balance}. If $j\notin H$, then the value of $\Upsilon_H^\I(\td\te)$ is fully determined by $h_{t,\I\setminus\{j\}}(\td\te)$, and $m^j(\td\te)$ by $h_{t,\emptyset}(\td\te)$. By Lemma \ref{lem:Z} and \eqref{eq:Z-conditional},
\begin{align}
    &\E_{\mu(h_{0,\I})}\lt[\Upsilon_H^\I Z^j\rt]=\E_{\mu(h_{0,\I})}\lt[\cE_{\mu(\eta_{t,\I\setminus\{j\}})}\lt[\Upsilon_H^\I Z^j\rt]\rt]=\E_{\mu(h_{0,\I})}\lt[\Upsilon_H^\I\cE_{\mu(\eta_{t,\I\setminus\{j\}})}\lt[Z^j\rt]\rt]=\E_{\mu(h_{0,\I})}\lt[\Upsilon_H^\I m^j\rt]\\
    &=\E_{\mu(h_{0,\I})}\lt[\cE_{\mu(\eta_{t,\emptyset})}\lt[\Upsilon_H^\I m^j\rt]\rt]=\E_{\mu(h_{0,\I})}\lt[\cE_{\mu(\eta_{t,\emptyset})}\lt[\Upsilon_H^\I\rt]m^j\rt]=\E_{\mu(h_{0,\I})}\lt[\Upsilon_\emptyset^\I m^j\rt].
\end{align}
Likewise, the contribution of the negative terms is
\begin{align}
    Q&=-\sum_{i\in\I}\E_{\mu(h_{0,\I})}\lt[\sum_{i\in H\subseteq\I}w_H\Upsilon_{H\setminus\{i\}}^\I Z^i\rt]=-\sum_{\emptyset\neq H\subseteq\I}\E_{\mu(h_{0,\I})}\lt[\sum_{i\in H}w_H\Upsilon_{H\setminus\{i\}}^\I Z^i\rt]\\
    &=-\sum_{\emptyset\neq H\subseteq\I}w_H\sum_{i\in H}\E_{\mu(h_{0,\I})}\lt[\Upsilon_\emptyset^\I m^i\rt].
\end{align}
Consequently,
\begin{align}
    P+Q=-\sum_{\emptyset\neq H\subseteq\I}w_H\sum_{i\in\I}\E_{\mu(h_{0,\I})}\lt[\Upsilon_\emptyset^\I m^i\rt]=-\sum_{\emptyset\neq H\subseteq\I}w_H\E_{\mu(h_{0,\I})}\lt[\Upsilon_\emptyset^\I\sum_{i\in\I} m^i\rt]=0,
\end{align}
because $\sum_i m^i=\cE_{\mu(\eta_{t,\emptyset})}[\sum_iZ^i]=0$. Summing over $t$ proves \eqref{eq:utility-orthogonality}.
\end{proof}

\section{Transfer Variance Minimizing GUE Implementation of truthful reporting}

In this section, we introduce the\footnote{See Proposition \ref{prop:uniqueness} about uniqueness.} GUE-implementation of truthful reporting that minimizes the total variance of total transfers. \textbf{To be clear, here we are talking about minimizing the variance of total transfers, not per period transfers.}

\subsection{A GUM-like ordered transfer rule}
Fix the order $0,1,\dots,N$ of agents, where $0$ denotes the public signal. For distinct agents $i$ and $j$,
\begin{align}
    \gamma'^{i\to j}_t=\Upsilon_{t,[i]\setminus\{j\}}^j-\Upsilon_{t,[i-1]\setminus\{j\}}^j.
\end{align}
The ordered transfer rule is
\begin{align}\label{eq:GUM-like transfers}
    y_t^i=\sum_{j\neq i}\lt(\gamma'^{i\to j}_t-\gamma'^{j\to i}_t\rt),\qquad t>0,
\end{align}
with $y_0^i=0$. The terms $\Upsilon_{t,[j]\setminus\{i\}}^i-\Upsilon_{t,[j-1]\setminus\{i\}}^i$ telescope, giving the alternate form
\begin{align}\label{eq:GUM-like alternate formula}
    y_t^i=(\Upsilon_{t,\emptyset}^i-\Upsilon_{t,\I\setminus\{i\}}^i)+\sum_{j\neq i}\gamma'^{i\to j}_t,\qquad t>0.
\end{align}

The interpretation is the same as for GUM, except that when the effect of $i$'s report on $j$ is calculated, $j$'s current report is omitted from both conditioning histories.

This GUM-like transfer rule is analogously defined for other agent orders.

\begin{Proposition}\label{prop:GUM-like is GUE}
    The transfer rule \eqref{eq:GUM-like transfers} is budget-balanced and is a GUE-implementation of truthful reporting.
\end{Proposition}
\begin{proof}
    Budget-balance follows immediately from the bilateral interpretation. We then verify \eqref{eq:elegseges} with the tower property, so by proposition \ref{elegseges}, the transfer rule is a GUE-implementation of truthful reporting. Appendix \ref{app:supplementary for GUM-like is GUE} records the calculation in expanded form.
\end{proof}

\subsection{The symmetrized rule}
For safety, assume $N\geq 2$.
 
Averaging \eqref{eq:GUM-like transfers} over all orders gives
\begin{align}\label{eq:sym-GUM-like}
    \overline{y}_t^i=\lt(\Upsilon_{t,\emptyset}^i-\Upsilon_{t,\I\setminus\{i\}}^i\rt)+\sum_{i\in H\subsetneq\I}\frac{\Upsilon_{t,H}^{\I\setminus H}-\Upsilon_{t,H\setminus\{i\}}^{\I\setminus H}}{|H|\binom{N-1}{|H|}},\qquad t>0,
\end{align}
with $\overline{y}_0^i=0$. This formula is derived in Appendix \ref{app:sym-GUM like derivation}. This is a GUE-implementation of truthful reporting as it is a convex combination of GUE-implementations of truthful reporting. The tower property gives
\begin{align}
    \E_{\mu(h_{0,\I})}\lt[\overline{Y}^i\rt]=0\qquad\text{for every }i.
\end{align}

\begin{Theorem}[Transfer variance minimization]\label{mintransfervariance}
    Suppose $N\geq 2$. For the fixed decision policy, the transfer rule \eqref{eq:sym-GUM-like} minimizes the total variance of \textbf{total transfers} among all GUE implementations of truthful reporting. We define the total variance of total transfers as $\sum_{i\in\I}\Var\lt(Y^i(\td\te)\rt)$, where $\td\te\sim\mu(h_{0,\I})$.
\end{Theorem}
\begin{proof}
Normalize an arbitrary competing truthful GUE rule $y'$ so that $\E_{\mu(h_{0,\I})}\lt[Y'^i\rt]=0$ for every $i$, exactly as in the proof of Theorem \ref{minutilityvariance}. Let
\begin{align}
    Z^i=Y'^i-\overline{Y}^i.
\end{align}
It is enough to prove
\begin{align}\label{eq:orthogonality-transfers}
    \sum_{i\in\I}\E_{\mu(h_{0,\I})}\lt[\overline{Y}^i Z^i\rt]=0,
\end{align}
because then
\begin{align}
    \sum_{i\in\I}\E_{\mu(h_{0,\I})}\lt[(Y'^i)^2\rt]=\sum_{i\in\I}\E_{\mu(h_{0,\I})}\lt[(\overline{Y}^i)^2\rt]+\sum_{i\in\I}\E_{\mu(h_{0,\I})}\lt[(Z^i)^2\rt].
\end{align}

Fix $t$ and write
\begin{align}
    v_H=\frac{1}{|H|\binom{N-1}{|H|}},\qquad m^i=\cE_{\mu(\eta_{t,\emptyset})}\lt[Z^i\rt],
\end{align}
and suppress the period index on $\Upsilon$. First consider the initial difference in \eqref{eq:sym-GUM-like}. Since both $\Upsilon_{\emptyset}^i(\td\te)$ and $\Upsilon_{\I\setminus\{i\}}^i(\td\te)$ are determined by $h_{t,\I\setminus\{i\}}(\td\te)$, Lemma \ref{lem:Z} and the tower property give
\begin{align}
    &\E_{\mu(h_{0,\I})}\lt[(\Upsilon_{\emptyset}^i-\Upsilon_{\I\setminus\{i\}}^i)Z^i\rt]=\E_{\mu(h_{0,\I})}\lt[\cE_{\mu(\eta_{t,\I\setminus\{i\}})}\lt[(\Upsilon_{\emptyset}^i-\Upsilon_{\I\setminus\{i\}}^i)Z^i\rt]\rt]\\
    &=\E_{\mu(h_{0,\I})}\lt[(\Upsilon_{\emptyset}^i-\Upsilon_{\I\setminus\{i\}}^i)\cE_{\mu(\eta_{t,\I\setminus\{i\}})}\lt[Z^i\rt]\rt]=\E_{\mu(h_{0,\I})}\lt[(\Upsilon_{\emptyset}^i-\Upsilon_{\I\setminus\{i\}}^i)m^i\rt]\\
    &=\E_{\mu(h_{0,\I})}\lt[\cE_{\mu(\eta_{t,\emptyset})}\lt[(\Upsilon_{\emptyset}^i-\Upsilon_{\I\setminus\{i\}}^i)m^i\rt]\rt]=\E_{\mu(h_{0,\I})}\lt[\cE_{\mu(\eta_{t,\emptyset})}\lt[(\Upsilon_{\emptyset}^i-\Upsilon_{\I\setminus\{i\}}^i)\rt]m^i\rt]=\E_{\mu(h_{0,\I})}\lt[0m^i\rt],
\end{align}
so this term contributes $0$ to \eqref{eq:orthogonality-transfers}. The contribution of the positive term of the remaining part is
\begin{align}
    P&=\sum_{i\in\I}\E_{\mu(h_{0,\I})}\lt[\sum_{i\in H\subsetneq\I}v_H\Upsilon_H^{\I\setminus H}Z^i\rt]=\sum_{\emptyset\neq H\subsetneq\I}\E_{\mu(h_{0,\I})}\lt[\sum_{i\in H}v_H\Upsilon_H^{\I\setminus H}Z^i\rt]\\
    &=-\sum_{\emptyset\neq H\subsetneq\I}v_H\sum_{j\notin H}\E_{\mu(h_{0,\I})}\lt[\Upsilon_{H}^{\I\setminus H} Z^j\rt]=-\sum_{\emptyset\neq H\subsetneq\I}v_H\sum_{j\notin H}\E_{\mu(h_{0,\I})}\lt[\Upsilon_{\emptyset}^{\I\setminus H}m^j\rt].
\end{align}
The third equality uses \eqref{eq:Z-balance} and the last uses \eqref{eq:Z-conditional} and the tower property. Similarly, the contribution of the negative term is
\begin{align}
    Q&=-\sum_{i\in\I}\E_{\mu(h_{0,\I})}\lt[\sum_{i\in H\subsetneq\I}v_H\Upsilon_{H\setminus\{i\}}^{\I\setminus H} Z^i\rt]=-\sum_{\emptyset\neq H\subsetneq\I}\E_{\mu(h_{0,\I})}\lt[\sum_{i\in H}v_H\Upsilon_{H\setminus\{i\}}^{\I\setminus H} Z^i\rt]\\
    &=-\sum_{\emptyset\neq H\subsetneq\I}v_H\sum_{i\in H}\E_{\mu(h_{0,\I})}\lt[\Upsilon_{\emptyset}^{\I\setminus H}m^i\rt].
\end{align}
Thus
\begin{align}
    P+Q=-\sum_{\emptyset\neq H\subsetneq\I}v_H\E_{\mu(h_{0,\I})}\lt[\Upsilon_{\emptyset}^{\I\setminus H}\sum_{i\in\I}m^i\rt]=0.
\end{align}
Summing over $t$ proves \eqref{eq:orthogonality-transfers}.
\end{proof}

\section{Uniqueness}
The inner product space structure of the two variance minimizer proofs immediately gives the following uniqueness result

\begin{Proposition}[Uniqueness of the truthful total transfers]\label{prop:uniqueness} After imposing the normalization $\E_{\mu(h_{0,\I})}\lt[Y^i\rt]=0$ for every agent, the minimizer in either Theorem \ref{minutilityvariance} or Theorem \ref{mintransfervariance} is unique in total transfers under truthful reporting, up to events of probability zero. Without normalization, another minimizer may differ by agent-specific constants $c^i$ satisfying $\sum_{i}c^i=0$.
\end{Proposition}
\begin{proof}
    For any competing normalized transfer rule, write $Z^i$ for the difference in total transfers between it and the respective variance minimizer rule (\eqref{modifiedsymGUM} or \eqref{eq:sym-GUM-like}) for agent $i$. The competing transfer rule is optimal exactly when $\sum_i\E_{\mu(h_{0,\I})}\lt[(Z^i)^2\rt]=0$, which holds if and only if $\E_{\mu(h_{0,\I})}\lt[(Z^i)^2\rt]=0$ for every $i$, which in turn holds if and only if $Z^i=0$ $\mu(h_{0,\I})$-a.s.

    Any non-normalized minimizer changes under normalization precisely by zero-sum constants, while the variances it induces remain unchanged. Therefore, such a minimizer can differ from a normalized minimizer only by zero-sum constant transfers.
\end{proof}

\appendix

\section{Supplementary calculations}

\subsection{Proof of Lemma \ref{lem:LOTE}}\label{app:LOTE}
In this section, we (mostly) suppress the decision policy $\x$ and coordinated strategy $\overline{s}$ from the notation, but this does not mean that we only consider the case of truthful reporting.

\begin{proof}
\begin{enumerate}[label=(\roman*)]
    \item If $g(\td\te)$ is determined by $h_{t,H}(\td\te)$, then under $\mu(h_{t,H}(\te))$, $g(\td\te)=g(\te)$ surely, so by linearity
    \begin{align}
        \cE_{\mu(\eta_{t,H})}\lt[fg\rt](\te)=\cE_{\mu(\eta_{t,H})}\lt[f\rt](\te)g(\te).
    \end{align}
    \item We distinguish three cases.   
\end{enumerate}
\begin{enumerate}[label=\alph*]
        \item$\eta_{t_1,H_1}\supseteq\eta_{t_2,H_2}$: Substituting $"f"\equiv 1$ and $"g"=\cE_{\mu(\eta_{t_2,H_2})}\lt[f\rt]$ into part (i) of this lemma gives
        \begin{align}
            \cE_{\mu(\eta_{t_1,H_1})}\lt[\cE_{\mu(\eta_{t_2,H_2})}\lt[f\rt]\rt]=\cE_{\mu(\eta_{t_1,H_1})}\lt[1\rt]\cE_{\mu(\eta_{t_2,H_2})}\lt[f\rt].
        \end{align}
        The expectation of $1$ is, of course, $1$, and the inclusion relation of the two history kernels gives $\eta_{t_2,H_2}=\eta_{t_1,H_1}\cap\eta_{t_2,H_2}$, which completes this case.
        \item$\eta_{t_1,H_1}\subseteq\eta_{t_2,H_2}$: Fix a $\te\in(\Te)^{T}$. We need to prove
        \begin{align}\label{eq:left}
            \E_{\mu(h_{t_1,H_1}(\te))}\lt[f\rt]=\E_{\td\te\sim\mu(h_{t_1,H_1}(\te))}\lt[\E_{\mu(h_{t_2,H_2}(\td\te))}\lt[f\rt]\rt].
        \end{align}
        Because the future cannot affect the past, and because $\E_{\mu(h_{t_2,H_2}(\td\te))}\lt[f\rt]$ is a function of $h_{t_2,H_2}(\td\te)$,
        \begin{align}\label{eq:right}
            \E_{\td\te\sim\mu(h_{t_1,H_1}(\te))}\lt[\E_{\mu(h_{t_2,H_2}(\td\te))}\lt[f\rt]\rt]=\E_{h_{t_2,H_2}(\td\te)\sim\mu(h_{t_1,H_1}(\te))}\lt[\E_{\mu(h_{t_2,H_2}(\td\te))}\lt[f\rt]\rt].
        \end{align}
        Now notice that the right side of \eqref{eq:right} is the same as the left side of \eqref{eq:left}, just rewritten as a double sum.
        \item Neither inclusion relation holds between $\eta_{t_1,H_1}$ and $\eta_{t_2,H_2}$: The only way this can occur is if $t_1=t_2\eqqcolon t>0$. Consequently, $\eta_{t,H_1}\cap\eta_{t,H_2}=\eta_{t,H_1\cap H_2}$. Write $g(h_{t,H_2}(\te^*))=\E_{\mu(h_{t,H_2}(\te^*))}\lt[f\rt]$ and $\nu^i(\te^*)=\mu^i({\te^*}_{t-1}^0,{\te^*}_{t-1}^i,\x(h_{t-1}(\h{\te^*})))$. Fix a $\te\in(\Te)^{T}$. Expanding the definition of the continuation law expectation operator, and considering that future events do not affect the distribution of current types nor the value of $g$ gives
        \begin{align}\label{eq:first expansion}
            \cE_{\mu(\eta_{t,H_1})}\lt[g(h_{t,H_2})(\td\te)\rt](\te)=\sum_{\td\te_t^{\I\setminus H_1}}\lt(\prod_{i\in\I\setminus H_1}\nu^i(\te)(\td\te_t^i)\rt)g(h_{t,H_1\cap H_2}(\te),\td\te_t^{H_2\setminus H_1})
        \end{align}
        and
        \begin{align}\label{eq:second expansion}
            \cE_{\mu(\eta_{t,H_1\cap H_2})}\lt[g(h_{t,H_2})(\td\te)\rt](\te)=\sum_{\td\te_t^{\I\setminus (H_1\cap H_2)}}\lt(\prod_{i\in\I\setminus (H_1\cap H_2)}\nu^i(\te)(\td\te_t^i)\rt)g(h_{t,H_1\cap H_2}(\te),\td\te_t^{H_2\setminus H_1}).
        \end{align}
        Conditional independence of the agent's current types lets us rearrange the sum in \eqref{eq:second expansion} into a double sum, then we can simplify:
        \begin{align}
            &\sum_{\td\te_t^{\I\setminus H_1}}\lt(\prod_{i\in\I\setminus H_1}\nu^i(\te)(\td\te_t^i)\rt)g(h_{t,H_1\cap H_2}(\te),\td\te_t^{H_2\setminus H_1})\sum_{\td\te_t^{H_1\setminus H_2}}\prod_{i\in H_1\setminus H_2}\nu^i(\te)(\td\te_t^i)\\
            &=\sum_{\td\te_t^{\I\setminus H_1}}\lt(\prod_{i\in\I\setminus H_1}\nu^i(\te)(\td\te_t^i)\rt)g(h_{t,H_1\cap H_2}(\te),\td\te_t^{H_2\setminus H_1})\cdot 1.
        \end{align}
        So \eqref{eq:first expansion} and \eqref{eq:second expansion} are equal. From this point, part (ii)/b completes the proof by
        \begin{align}
            \cE_{\mu(\eta_{t,H_1\cap H_2})}\lt[\cE_{\mu(\eta_{t,H_2})}\lt[f\rt]\rt]=\cE_{\mu(\eta_{t,H_1\cap H_2})}\lt[f\rt].
        \end{align}
    \end{enumerate}
\end{proof}

\subsection{Proof of Lemma \ref{utilitysum}}\label{app:utilitysum}
\begin{proof}
    Fix a (pure) coordinated strategy $\bar s\in\bar\S^\I$, and for an arbitrary $\te^*\in(\Te)^{T}$, let $\h{\te^*}$ denote the reports generated recursively from it by $\overline{s}$. Because the transfer rule is budget-balanced, the total utility of the agents equals their total payoff. Consequently, we need only show
    \begin{align}
        \sum_{i\in\I}\E_{\mu(h_{0,\I},\mathsf{x},\overline{s})}\lt[\sum_{t=0}^Tu^i(\td\te_t^i,\mathsf{x}(h_t(\h{\td\te})))\rt]\leq\sum_{i\in\I}\E_{\mu(h_{0,\I},\mathsf{x})}\lt[\sum_{t=0}^Tu^i(\td\te_t^i,\mathsf{x}(h_{t}(\td\te)))\rt].
    \end{align}
    Define the decision policy $\mathsf{x}'$ by
\[
\mathsf x'(h_t(\theta^*))
=
\mathsf x\bigl(h_t(\widehat{\theta^*})\bigr).
\]
Non-anticipativity of $\bar s$ makes $\mathsf x'$ a valid decision policy. It is clear that $\mu(h_{0,\I},\mathsf{x},\overline s)=\mu(h_{0,\I},\mathsf{x}')$, considering that for any state profile $\te\in(\Te)^{T}$ and period $t$, $\mathsf{x}'(h_{t}(\te))=\mathsf{x}(h_{t}(\h\te))$. Consequently,
    \begin{align}        
    \E_{\mu(h_{0,\I},\mathsf{x},\bar s)}\lt[\sum_{i\in\I}\sum_{t=0}^Tu^i(\td\te_t^i,\mathsf{x}(h_t(\h{\td\te})))\rt]=\E_{\mu(h_{0,\I},\mathsf{x}')}\lt[\sum_{i\in\I}\sum_{t=0}^Tu^i(\td\te_t^i,\mathsf{x}'(h_t(\td\te)))\rt].
    \end{align}
    By the choice of $\mathsf{x}$ shown in equation \eqref{thedecisionpolicy},
    \begin{align}
        \E_{\mu(h_{0,\I},\mathsf{x}')}\lt[\sum_{i\in\I}\sum_{t=0}^Tu^i(\td\te_t^i,\mathsf{x}'(h_t(\td\te)))\rt]\leq\E_{\mu(h_{0,\I},\mathsf{x})}\lt[\sum_{i\in\I}\sum_{t=0}^Tu^i(\td\te_t^i,\mathsf{x}(h_t(\td\te)))\rt].
    \end{align}
\end{proof}

\subsection{Deriving the formula of sym-GUM \eqref{sym-GUM formula}}\label{GUM to sym-GUM}
    In this section, $\overline y$ denotes the sym-GUM transfer rule \eqref{sym-GUM formula}.

    The formula of $\overline{y}_t^i$ ($t>0$) is a linear combination of the terms $\Upsilon_{t,\emptyset}^i$, $\Upsilon_{t,\I}^i$, and terms of the form $\Upsilon_{t,H}^\I-\Upsilon_{t,H\setminus\{i\}}^\I$, where $i\in H\subseteq\I$, because for every order of agents the GUM transfer \eqref{eq:GUM alterante formula} to agent $i$ in period $t>0$ is a linear combination of such terms, and sym-GUM is a linear combination of ordered GUM transfers.

    The coefficients of $\Upsilon_{t,\emptyset}^i$ and $\Upsilon_{t,\I}^i$ in the formula of $\overline{y}_t^i$ are $1$ and $-1$ respectively, as they have these coefficients in every ordered GUM transfer rule. For a fixed subset of agents $H\subseteq\I$ that contains $i$, the coefficient of the term $\Upsilon_{t,H}^\I-\Upsilon_{t,H\setminus\{i\}}^\I$ will be $\frac{1}{N!}$ times the number of orders of agents in which $i$ is the $|H|$-th agent, and the agents in $H\setminus\{i\}$ are exactly the ones to come before $i$ in the order. This comes out to be
    \begin{align}
        \frac{1}{N!}\cdot\lt((|H|-1)!\cdot(N-|H|)!\rt)=\frac{1}{|H|\binom{N}{|H|}}.
    \end{align}
    Putting it all together, we get that
    \begin{align}
        \overline y_t^i=\Upsilon_{t,\emptyset}^i-\Upsilon_{t,\I}^i+\sum_{i\in H\subseteq\I}\frac{\Upsilon_{t,H}^\I-\Upsilon_{t,H\setminus\{i\}}^\I}{|H|\binom{N}{|H|}},
    \end{align}
    where $t>0$. The transfers in period $0$ of course remain $0$.

    \subsection{Deriving the formula of the symmetrized GUM-like transfer rule \eqref{eq:sym-GUM-like}}\label{app:sym-GUM like derivation}
    In this section, $\overline{y}$ denotes the symmetrized version of the GUM-like transfer rule \eqref{eq:GUM-like transfers}.

    The formula of $\overline y_t^i$ ($t>0$) is a linear combination of the terms $\Upsilon_{t,\emptyset}^i$ and $\Upsilon_{t,\I\setminus\{i\}}^i$, and terms of the form $\Upsilon_{t,H}^j-\Upsilon_{t,H\setminus\{i\}}^j$, where $i\neq j\in\I$, $i\in H\subseteq\I$, and $j\notin H$, because for every order of agents, the GUM-like transfer rule \eqref{eq:GUM-like alternate formula} is a linear combination of such terms, and $\overline{y}$ is a linear combination of these ordered transfer rules.

    The coefficients of $\Upsilon_{t,\emptyset}^i$ and $\Upsilon_{t,\I\setminus\{i\}}^i$ in the formula of $\overline{y}_t^i$ are $1$ and $-1$ respectively, since they have these coefficients in the formula of the GUM-like transfers regardless of agent order. For a fixed agent $j\in\I\setminus\{i\}$ and a fixed subset of agents $H\subseteq\I$ that contains $i$ and does not contain $j$, the coefficient of the term $\Upsilon_{t,H}^j-\Upsilon_{t,H\setminus\{i\}}^j$ is $\frac{1}{N!}$ times the number of orders of agents where either the agents to come before $i$ in the order are exactly the ones in $H\setminus\{i\}$, or the agents to come before $i$ are exactly the ones in $(H\setminus\{i\})\cup\{j\}$. The coefficient is then
    \begin{align}
        \frac{1}{N!}\lt((|H|-1)!\cdot(N-|H|)!+|H|!\cdot(N-|H|-1)!\rt)=\frac{1}{|H|\binom{N-1}{|H|}}.
    \end{align}
    Putting it all together, we get that
    \begin{align}
        \overline y_t^i=&\Upsilon_{t,\emptyset}^i-\Upsilon_{t,\I\setminus\{i\}}^i+\sum_{i\neq j\in\I}\sum_{i\in H\subseteq\I\setminus\{j\}}\frac{\Upsilon_{t,H}^j-\Upsilon_{t,H\setminus\{i\}}^j}{|H|\binom{N-1}{|H|}}\\
        =&\Upsilon_{t,\emptyset}^i-\Upsilon_{t,\I\setminus\{i\}}^i+\sum_{i\in H\subsetneq\I}\sum_{j\in\I\setminus H}\frac{\Upsilon_{t,H}^j-\Upsilon_{t,H\setminus\{i\}}^j}{|H|\binom{N-1}{|H|}}\\
        =&\Upsilon_{t,\emptyset}^i-\Upsilon_{t,\I\setminus\{i\}}^i+\sum_{i\in H\subsetneq\I}\frac{\Upsilon_{t,H}^{\I\setminus H}-\Upsilon_{t,H\setminus\{i\}}^{\I\setminus H}}{|H|\binom{N-1}{|H|}},
    \end{align}
    where $t>0$. The transfers in period $0$ remain $0$.

\subsection{Supplementary calculations for Proposition \ref{prop:GUM-like is GUE}}\label{app:supplementary for GUM-like is GUE}
Let $y$ denote the GUM-like transfer rule \eqref{eq:GUM-like transfers}.

Fix an agent $i$ and a period $t>0$. For $\tau<t$,
\begin{align}
    \cE_{\mu(\eta_{t,\I\setminus\{i\}})}\lt[y_\tau^i\rt]=\cE_{\mu(\eta_{t,\emptyset})}\lt[y_\tau^i\rt]=y_\tau^i,
\end{align}
because $y_\tau^i(\td\te)$ is determined by the history $h_{t,\emptyset}(\td\te)$. For $\tau>t$, by the law of total expectation
\begin{align}
    \cE_{\mu(\eta_{t,\I\setminus\{i\}})}\lt[y_\tau^i\rt]=\lt(\Upsilon_{t,\I\setminus\{i\}}^i-\Upsilon_{t,\I\setminus\{i\}}^i\rt)+\sum_{j\neq i}\lt(\Upsilon_{t,\I\setminus\{i\}}^j-\Upsilon_{t,\I\setminus\{i\}}^j\rt)=0
\end{align}
and
\begin{align}
    \cE_{\mu(\eta_{t,\emptyset})}\lt[y_\tau^i\rt]=\lt(\Upsilon_{t,\emptyset}^i-\Upsilon_{t,\emptyset}^i\rt)+\sum_{i\neq j}\lt(\Upsilon_{t,\emptyset}^j-\Upsilon_{t,\emptyset}^j\rt)=0.
\end{align}
Finally, by the law of total expectation
\begin{align}
    \cE_{\mu(\eta_{t,\I\setminus\{i\}})}\lt[\Upsilon_{T,\I}^i+y_t^i\rt]=\Upsilon_{t,\I\setminus\{i\}}^i+\lt(\Upsilon_{t,\emptyset}^i-\Upsilon_{t,\I\setminus\{i\}}^i\rt)+\sum_{j\neq  i}\lt(\Upsilon_{t,[i-1]\setminus\{j\}}^j-\Upsilon_{t,[i-1]\setminus\{j\}}^j\rt)=\Upsilon_{t,\emptyset}^i
\end{align}
and
\begin{align}
    \cE_{\mu(\eta_{t,\emptyset})}\lt[\Upsilon_{T,\I}^i+y_t^i\rt]=\Upsilon_{t,\emptyset}^i+\lt(\Upsilon_{t,\emptyset}^i-\Upsilon_{t,\emptyset}^i\rt)+\sum_{j\neq i}\lt(\Upsilon_{t,\emptyset}^j-\Upsilon_{t,\emptyset}^j\rt)=\Upsilon_{t,\emptyset}^i.
\end{align}
So the budget-balanced GUM-like transfer rule \eqref{eq:GUM-like transfers} satisfies \eqref{eq:elegseges}, thus it is a GUE-implementation of truthful reporting.

\subsection{Supplementary calculations for Note \ref{GUMisGUE}}\label{app:GUMisGUEcalculation}
Let $y$ denote the GUM transfer rule \eqref{eq:GUM transfers}.

Fix an agent $i$ and a period $t>0$. For $\tau<t$,
\begin{align}
    \cE_{\mu(\eta_{t,\I\setminus\{i\}})}\lt[y_\tau^i\rt]=\cE_{\mu(\eta_{t,\emptyset})}\lt[y_\tau^i\rt]=y_\tau^i,
\end{align}
because $y_\tau^i(\td\te)$ is determined by the history $h_{t,\emptyset}(\td\te)$. For $\tau>t$, by the law of total expectation
\begin{align}
    \cE_{\mu(\eta_{t,\I\setminus\{i\}})}\lt[y_\tau^i\rt]=\lt(\Upsilon_{t,\I\setminus\{i\}}^i-\Upsilon_{t,\I\setminus\{i\}}^i\rt)+\lt(\Upsilon_{t,\I\setminus\{i\}}^\I-\Upsilon_{t,\I\setminus\{i\}}^\I\rt)=0
\end{align}
and
\begin{align}
    \cE_{\mu(\eta_{t,\emptyset})}\lt[y_\tau^i\rt]=\lt(\Upsilon_{t,\emptyset}^i-\Upsilon_{t,\emptyset}^i\rt)+\lt(\Upsilon_{t,\emptyset}^\I-\Upsilon_{t,\emptyset}^\I\rt)=0.
\end{align}
Finally, by the law of total expectation
\begin{align}
    \cE_{\mu(\eta_{t,\I\setminus\{i\}})}\lt[\Upsilon_{T,\I}^i+y_t^i\rt]=\Upsilon_{t,\I\setminus\{i\}}^i+\lt(\Upsilon_{t,\emptyset}^i-\Upsilon_{t,\I\setminus\{i\}}^i\rt)+\lt(\Upsilon_{t,[i-1]}^\I-\Upsilon_{t,[i-1]}^\I\rt)=\Upsilon_{t,\emptyset}^i
\end{align}
and
\begin{align}
    \cE_{\mu(\eta_{t,\emptyset})}\lt[\Upsilon_{T,\I}^i+y_t^i\rt]=\Upsilon_{t,\emptyset}^i+\lt(\Upsilon_{t,\emptyset}^i-\Upsilon_{t,\emptyset}^i\rt)+\lt(\Upsilon_{t,\emptyset}^\I-\Upsilon_{t,\emptyset}^\I\rt)=\Upsilon_{t,\emptyset}^i.
\end{align}
So the budget-balanced GUM transfer rule \eqref{eq:GUM transfers} satisfies \eqref{eq:elegseges}, thus it is a GUE-implementation of truthful reporting.

\vspace{-0.5cm}

\singlespacing
\bibliographystyle{chicago}
\bibliography{Literature}

\end{document}